\documentclass{llncs}
\usepackage{graphicx}
\usepackage{caption}
\usepackage{subcaption}
\usepackage{amssymb}
\usepackage{amsmath}
\usepackage{amsfonts}
\usepackage[english]{babel}
\usepackage[bottom]{footmisc}
\usepackage{wrapfig}
\begin{document}

\newtheorem{mydef}{Definition}
\title{Stochastic Analysis of Chemical Reaction Networks Using Linear Noise Approximation\thanks{This research is supported by a Royal Society Research Professorship and ERC AdG VERIWARE.}}

\author{Luca Cardelli\inst{1,2} \and Marta Kwiatkowska\inst{2}  \and Luca Laurenti\inst{2} }

\institute{Microsoft Research 
\and Department of Computer Science, University of Oxford }

\maketitle

\begin{abstract}
Stochastic evolution of Chemical Reactions Networks (CRNs) over time is usually analysed through solving the Chemical Master Equation (CME) or performing extensive simulations. Analysing stochasticity is often needed, particularly when some molecules occur in low numbers.
Unfortunately, both approaches become infeasible if the system is complex and/or it cannot be ensured that initial populations are small.
We develop a probabilistic logic for CRNs that enables stochastic analysis of the evolution of populations of molecular species. We present an approximate model checking algorithm based on the Linear Noise Approximation (LNA) of the CME, whose computational complexity is independent of the population size of each species and polynomial in the number of different species. The algorithm requires the solution of first order polynomial differential equations. We prove that our approach is valid for any CRN close enough to the thermodynamical limit.
However, we show on four case studies that it can still provide good approximation even for low molecule counts. Our approach enables rigorous analysis of CRNs that are not analyzable by solving the CME, but are far from the deterministic limit. Moreover, it can be used for a fast approximate stochastic characterization of a CRN.

\end{abstract}
\section{Introduction}
Chemical reaction networks (CRNs) and mass action kinetics are well studied formalisms for modelling biochemical systems \cite{Chellaboina2009}. In recent years, CRNs have also been successfully used as a formal programming language for biochemical systems \cite{Soloveichik2009,Cardelli2010,Chen2013}.
There are two well established approaches for analyzing chemical networks: deterministic and stochastic \cite{Gillespie2013}. The deterministic approach models the kinetics of a CRN as a system of ordinary differential equations (ODEs) and represents average behaviour, valid in the thermodynamic limit 
\cite{Gillespie2009}. The stochastic approach, on the other hand, is based on the Chemical Master Equation (CME) and models the CRN as a continuous-time Markov chain (CTMC) \cite{Cardelli2008}. The stochastic behavior can be analyzed by stochastic simulation \cite{Gillespie2013} or by exhaustive probabilistic model checking  of the CTMC, which can be performed, for example, by using PRISM \cite{KNP11}.

Exhaustive analysis of the CTMC is able to find the best- and worst-case scenarios and is correct for any population size, but suffers from the state-space explosion problem \cite{kwiatkowska2014probabilistic} and can only be used for relatively small systems. In contrast, deterministic methods are much more robust with respect to state-space explosion, but unable to represent stochastic fluctuations, which play a fundamental role when the system is not in thermodynamic equilibrium.

\textbf{Contributions.}
In this paper we develop a novel approach for analysing the stochastic evolution of a CRN based on the Linear Noise Approximation (LNA) of the CME. We formulate SEL (Stochastic Evolution Logic), a probabilistic logic for CRNs that enables reasoning about probability, expectation and variance of linear combinations of populations of the species. Examples of properties that can be specified in our logic are shown in Example \ref{ex-1}. We propose an approximate model checking algorithm for the logic based on the LNA and implement it in Matlab and Java. We demonstrate that the complexity of model checking is polynomial in the initial number of species and independent of the initial molecule counts, thus ameliorating state-space explosion. Further, we show that model checking is exact when approaching the thermodynamic limit. 
Though the algorithm may not be accurate for systems far from the deterministic limit, this generally happens when the populations are small, in which case the analysis can be performed by transient analysis of the induced CTMC \cite{Kwiatkowska2007}. Our approach is essential for CRNs that cannot be analyzed by (partial) state space exploration, because of large or infinite state spaces. Moreover, it is useful for a fast (approximate) stochastic characterization of CRNs, since solving the LNA is much faster than solving the CME \cite{Elf2003}. 
We prove asymptotic correctness of LNA-based model checking and show on four examples that it is still possible to obtain very good approximations even for small population systems, comparing with standard uniformisation \cite{Kwiatkowska2007} and statistical model checking implemented in PRISM~\cite{KNP11}.


\textbf{ Related work.}
The closest work to ours is by Bortolussi \emph{et al.} \cite{Bortolussi2013}, which uses the Central Limit Approximation (CLA) (essentially the same as the LNA) for checking restricted timed automata specifications, assuming a fixed population size.
Wolf \emph{et al.} \cite{Wolf2010} develop a sliding window method to approximately verify infinite-state CTMCs, which applies to cases where most of the probability mass is concentrated in a confined region of the state space. 
Recently, Finite State Projection algorithms (FSP algorithms) for the solution or approximation of the CME have been introduced \cite{munsky2006finite}. Both methods apply to the induced CTMC, but require at least partial exploration of the state space, and are thus not immune to state-space explosion.

\textbf{ Structure of the paper.}
In Section \ref{crn-sec} we summarise the deterministic and stochastic modelling approaches for CRNs, and in Section \ref{lna-sec} we describe the Linear Noise Approximation method. 
Section \ref{logic-sec} introduces the logic SEL and the corresponding model checking algorithm based on the LNA. In Section \ref{results-sec} we demonstrate our approach on four networks taken from the literature. Section~\ref{concl-sec} concludes the paper.

\section{Chemical Reaction Networks}\label{crn-sec}
A \emph{chemical reaction network (CRN)} $C=(\Lambda,R)$ is a pair of finite sets, where $\Lambda$ is the set of \emph{chemical species} and $R$ the set of reactions.
$|\Lambda|$ denotes the size of the set of species. A \emph{reaction} $\tau \in R$ is a triple $\tau=(r_{\tau},p_{\tau},k_{\tau})$, where $r_{\tau},p_{\tau} \in  \mathbb{N}^{|\Lambda|}$ and $k_{\tau} \in \mathbb{R}_{>0} $. $r_{\tau}$ and $p_{\tau}$ represent the stoichiometry of reactants and products and $k_{\tau}$ is the coefficient associated to the rate of the reaction; its dimension is $s^{-1}$. 
We often write reactions as $\lambda_1 + \lambda_3 \, \rightarrow^{k_1}  \,    2\lambda_2 $ instead of $\tau_1=(  [1,0,1]^T,[0,2,0]^T,k_1 )$, where $\cdot^T$ indicates the transpose of a vector.
We define the \emph{net change} associated to a reaction $\tau$ by $\upsilon_{\tau}=p_{\tau} - r_{\tau}$.  For example, for $\tau_1$ as above, we have $\upsilon_{\tau_1}=[-1,2,-1]^T$.

We make the assumption that the system is well stirred, that is, the probability of the next reaction occurring between two molecules is independent of the location of those molecules. We consider fixed volume $V$ and temperature; under these assumptions a \emph{configuration} or \emph{state} $x \in \mathbb{N}^{|\Lambda|}$ of the system is given by the number of molecules of each species. 
We define $[x]=\frac{x}{N}$, the vector of the species \emph{concentration} in $x$ for a given $N$, where $N=V \cdot N_A$ is the volumetric factor, $V$ is the volume of the solution and $N_A$ is Avogadro's number. The physical dimension of $N$ is $Mol^{-1} \cdot L$, where $Mol$ indicates mole and $L$ is litre. Given $\lambda_i \in \Lambda$ then $\#\lambda_i\_x \in \mathbb{N}$ represents the number of molecules of $\lambda_i$ in $x$ and $[\lambda_i]\_x\in \mathbb{R}$ the concentration of $\lambda_i$ in the same configuration. In some cases we elide $x$, and we simply write $\#\lambda_i$ and $[\lambda_i]$ instead of $\#\lambda_i\_x$ and $[\lambda_i]\_x$. They are related by $[\lambda_i]=\frac{\#\lambda_i}{N}$. The dimension of $[\lambda_i]$ is $Mol \cdot L^{-1}$.

The propensity $\alpha_{n,\tau}$ of a reaction $\tau$ in terms of the number of molecules is a function of the current configuration of the system $x$ such that $\alpha_{n,\tau}(x)dt$ is the probability that a reaction event   occurs in the next infinitesimal interval $dt$.
In this paper we assume as valid the stochastic form of the law of mass action, so the propensity rates are proportional to the number of molecules that participate in the reaction \cite{Cardelli2008}. 
Stochastic models consider the system in terms of numbers of molecules, while deterministic ones, generally, in terms of concentrations, 
and the relationship is as follows. 
For a reaction $\tau=(r_{\tau},p_{\tau},k_{\tau})$, given the configuration $x$ and ${r_{\tau,i}}$, the $i$-th component of $r_{\tau}$, then $\alpha_{c,\tau}(x)=k_{\tau} \prod_{i=1}^{|\Lambda|}{([\lambda_i]\_x)}^{r_{\tau,i}}$ is the propensity function expressed in terms of concentrations as given by the deterministic law of mass action.
It is possible to show that, for any order of reaction, $\alpha_{n,\tau}(x)\approx  N\alpha_{c,\tau}(x)$ if $N$ is sufficiently large \cite{Anderson2011}. Note that $\alpha_{c,\tau}$ is independent of $N$.
In this paper we are interested only in finite time horizon, because 
of the problematic character of studying solutions of ODEs for infinite time horizon \cite{bortolussi2013continuous}.

\begin{wrapfigure}{r}{0.46\textwidth}
  \vspace{-3em}
  \begin{center}
\includegraphics[scale=0.34]{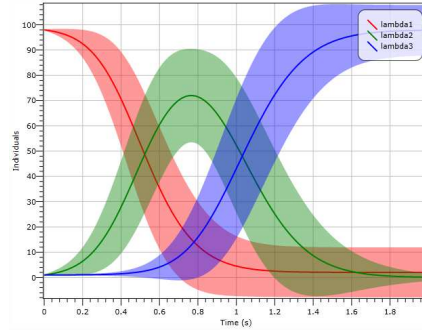}
  \end{center}
    \vspace{-1em}
  \caption{Expected number and standard deviation of species of the CRN of Example \ref{ex-1} for the given initial conditions, calculated by simulating the CME.}
    \label{fig:Example}
   \vspace{-3em}
\end{wrapfigure}
\begin{example}\label{ex-1}
Consider the CRN $C=(\{\lambda_1,\lambda_2,\lambda_3\},R)$, where $R=\{ (\lambda_1 + \lambda_2 \rightarrow^{10} \lambda_2 + \lambda_2),
(\lambda_2 + \lambda_3 \rightarrow^{10} \lambda_3 + \lambda_3) \}$, with initial conditions $\#\lambda_1=98,\#\lambda_2=1,\#\lambda_3=1$, for a system with $N=1000$. 
Figure \ref{fig:Example} plots the expectation and standard deviation of population sizes. We may wish to check if the maximum expected value of $\#\lambda_2$ remains smaller than $75$ molecules during the first $2 sec$. However, the system is stochastic, so we also need to analyse whether the variance is limited enough when $\#\lambda_2$ reaches the maximum. Sometimes, analysis of first and second moments does not suffice, so it could be of interest to check the probability of some events, for instance, is the probability that, between $t_1=0.5 sec$ and $t_2=1.0 sec$, $\#\lambda_2-(\#\lambda_1+\#\lambda_3)>0$ greater than $0.6?$ 
\end{example}

\textbf{Deterministic semantics.}
Let $C=(\Lambda,R)$ be a CRN. The deterministic model approximates the concentration of the species of the system over time as a set of autonomous polynomial first order differential equations:
\begin{equation}
 \frac{\mathrm d \Phi(t)}{\mathrm d t} = F(\Phi(t))
\label{eq:ODE}
\end{equation}
$F(\Phi(t))=\sum_{\tau=(r_{\tau},p_{\tau},k_{\tau}) \in R} \upsilon_{\tau} \alpha_{c,\tau}(\Phi(t))$ and $\alpha_{c,\tau}(\Phi(t))=k_{\tau}\prod_{i=1}^{|\Lambda|}{\Phi_i(t)}^{r_{\tau,i}}$. Function $\Phi : \mathbb{R}_{\geq 0} \rightarrow \mathbb{R}^{|\Lambda|} $ describes the behaviour of the system as a set of deterministic equations assuming a continuous state-space semantics, therefore $\Phi(t) \in \mathbb{R}^{|\Lambda|}$ is the vector of the species concentrations at time $t$. Assuming $t_0=0$, the initial condition is $\Phi(0)=[x_0]$, expressed as a concentration.  
 Note that $F(\Phi(t))$ is Lipschitz continuous, so $\Phi$ exists and is unique \cite{ethier2009markov}.

\textbf{Stochastic semantics.}%
CRNs are well represented by CTMCs, whose transient analysis can be performed via the Chemical Master Equation (CME) \cite{Kampen1992b}.
\begin{mydef}\label{ctmc-defn}
Given a CRN $C=(\Lambda,R)$ and the volumetric factor $N$, we define a time-homogeneous CTMC \cite{cinlar2013introduction,pinsky2010introduction}  $(X^N(t),t \in \mathbb{R}_{\geq 0})$ with state space $S = {\mathbb{N}}^{|\Lambda|}$. Given $x_0 \in S$, the initial configuration of the system, then  $P(X^N(0)=x_0)=1$.
The transition rate from state $x_i$ to state $x_j$ is defined as $r(x_i,x_j) =\sum_{\{\tau \in R | x_j=x_i+ v_{\tau}\}}  N\alpha_{c,\tau}(x_i)$.
\end{mydef}
$X^N(t)$ describes the stochastic evolution of molecular populations of each species at time $t$. For $x \in S$, we define $P^{(t)}(x)=P(X^N(t)=x|X(0)=x_0)$, where $x_0$ is the initial configuration. The CME describes the time evolution of $X^N$ as:
\begin{equation}\begin{split}
	&\frac{\mathrm d}{\mathrm d t} \left( P^{(t)}(x)\ \right) = 
	\sum_{\tau \in R} \{ N\alpha_{c,\tau}(x-\upsilon_{\tau})P^{(t)}(x-\upsilon_{\tau})-N\alpha_{c,\tau}(x)P^{(t)}(x) .\} \\
	\end{split}
\label{eq:CME}
\end{equation}
The CME can be equivalently defined in terms of the infinitesimal generator matrix \cite{Wolf2010}, which admits computing an approximation of the CME using, for example, fast adaptive uniformisation \cite{didier2009fast,DHK14} or the sliding window method \cite{Wolf2010}.

We also define the CTMC $(\frac{X^N(t)}{N},t \in \mathbb{R}_{\geq 0})$ with state space $S =\mathbb{Q}^{|\Lambda|}$. If $[x_0] \in S$ is the initial configuration, then  $P(\frac{X^N(0)}{N}=[x_0])=1.$
The transition rate from state $[x_i]$ to $[x_j]$ is defined as $r([x_i],[x_j]) =\sum_{\{\tau \in R  | [x_j]=[x_i]+ \frac{v_{\tau}}{N}\}}  N\alpha_{c,\tau}(x_i)$. $\frac{X^N(t)}{N}$ is the random vector describing the system at time $t$ in terms of concentrations.
In \cite{Anderson2011,ethier2009markov} it is proved that $\lim\limits_{N \rightarrow \infty}  \sup\limits_{t'\leq t} \| \frac{X^N(t')}{N} -\Phi(t')] \|=0 $ almost surely for every time $t$. This explains the relationship between the two different semantics, where the deterministic solution can be viewed as a limit of the stochastic solution, valid when close enough to the thermodynamic limit.

\section{Linear Noise Approximation}\label{lna-sec}
The solution of the CME can be computationally expensive, or even infeasible, because the set of reachable states can be huge or infinite. The Linear Noise Approximation (LNA) has been introduced by Van Kampen as a second order approximation of the system size expansion of the CME \cite{Kampen1992b}. 
Since stochastic fluctuations depend on $N$, and specifically, for average concentrations, are of the order of $N^{\frac{1}{2}}$ \cite{Elf2003,pinsky2010introduction}, to derive the expansion Van Kampen assumes that:
\begin{equation}
	 X^{N}(t) \approx N\Phi(t) + N^{\frac{1}{2}}Z(t)
\label{eq:hypothsis}
\end{equation}
 where $Z(t)=(Z_1(t),Z_2(t),...,Z_{|\Lambda|})$ is the random vector, independent of $N$, representing the stochastic fluctuations, $\Phi(t)$ is given by the solution of Eqn \eqref{eq:ODE} and  $X^N(t)$ is the random vector of Definition ~\ref{ctmc-defn}. Using this substitution in the system size expansion and then truncating at the second order, the probability distribution of $Z(t)$ is found to be given by the following linear Fokker-Plank equation \cite{Elf2003}:
\begin{equation}
\frac{\mathrm \partial P(Z,t) }{\mathrm \partial t} =-\sum\limits_{i=1}^{|\Lambda|} \sum\limits_{j=1}^{|\Lambda|}
\frac{\mathrm \partial F_j(\Phi(t))}{\mathrm \partial \Phi_i}  \frac{\mathrm \partial (Z_j P(Z,t)) }{\mathrm \partial Z_i} + \frac{1}{2} \sum\limits_{i=1}^{|\Lambda|} \sum\limits_{j=1}^{|\Lambda|} G_{i,j}(\Phi(t))  \frac{\mathrm \partial^2 P(Z,t) }{\mathrm \partial Z_i \partial Z_j}
\label{eq:FokkerPlanck}
\end{equation}
where $ G(\Phi(t))= \sum_{\tau \in R} \upsilon_{\tau} {\upsilon_{\tau}}^T \alpha_{c,\tau}(\Phi(t)) $ and $F_j(\Phi(t))$ is the $j-$th component of $F(\Phi(t))$. The solution of Eqn \eqref{eq:FokkerPlanck} gives a Gaussian process. For every time $t$, $Z(t)$ has a multivariate normal distribution, whose expected value and covariance matrix are the solution of the following equations \cite{Elf2003,Goutsias2013}:
\begin{equation}
		\frac{\mathrm d E[Z(t)]}{\mathrm d t}  = J_F(\Phi(t))E[Z(t)]
\label{eq:EV1}
\end{equation}
\begin{equation}
		\frac{\mathrm d C[Z(t)] }{\mathrm d t}  = J_F(\Phi(t))C[Z(t)] + C[Z(t)]{J^T}_F(\Phi(t))+G(\Phi(t))
\label{eq:COV1}
\end{equation}
where ${J}_F(\Phi(t))$ is the Jacobian of $F(\Phi(t))$. We consider as initial conditions $E[Z(0)]=0$ and $C[Z(0)]=0$. This means that $E[Z(t)]=0$ for every $t$.

It is possible to justify the hypothesis \eqref{eq:hypothsis} noting that in the lowest order the CME expansion reduces to Eqn \eqref{eq:ODE}, and with the following theorem by Kurtz:
\begin{theorem}{\cite{ethier2009markov}}\label{lna-thm}
Consider the subset $E \subset \mathbb{R}^{|\Lambda|}$ on which are defined the propensity functions $\alpha_{c,\tau}$. Let $Z^N(t)$ be the random vector given by $ Z^N(t)={N}^{\frac{1}{2}}(\frac{X^{N}(t)}{N}-\Phi(t))$. Suppose that $\sum\limits_{\tau \in R}|{v_{\tau}}^2| \sup\limits_{X\in K} \alpha_{c,\tau}(X)< \infty$ for each compact $K \subset E$, and that, for $N \rightarrow \infty$, $Z^N(0)=Z(0)$, then $Z^N(t)$ converges in distribution to $Z(t)$.
\end{theorem}
The LNA thus permits approximation of the probability distribution of $X^N(t)$ with the probability distribution of $Y^N(t)=N\Phi(t) + N^{\frac{1}{2}}Z(t)$. It is easy to show that $Y^N(t)$ has a Gaussian distribution; indeed, $Z(t)$ is Gaussian distributed, and $N$ and $\Phi(t)$ are deterministic.

To compute the LNA it is necessary to solve $O(|\Lambda|^2)$ first order differential equations, but the complexity is independent of the initial number of molecules of each species. 
Therefore, one can avoid the exploration of the state space that methods based on uniformisation rely upon. 

Theorem \ref{lna-thm} alone only guarantees convergence in distribution. However, in \cite{Wallace2012}, LNA is derived as an approximation of the Chemical Langevin Equation (CLE) \cite{Gillespie2000}, rather than system size expansion. This shows that LNA is valid for every real chemical system close enough to the thermodynamical limit, at least for a limited time. 
Thus, LNA is exact in the limit of high populations, but can also be used for small populations if the behaviour is not too far from the deterministic limit, taking into account the continuous nature of the approximation and Gaussian assumptions on the noise \cite{Goutsias2013,Wallace2012}. 

\subsection{Probabilistic analysis of CRNs}\label{proba-sec}

We have shown that $X^N$ can be approximated by $Y^{N}(t)=N\Phi(t) + N^{\frac{1}{2}}Z(t)$, where $Y^N(t)$ has a multivariate Gaussian distribution, 
%
so it is completely characterized by its expected value and covariance matrix, whose values are respectively $E[Y^{N}(t)]=N\Phi(t)$ and $C[Y^{N}(t)]=N^{\frac{1}{2}} C[Z(t)] N^{\frac{1}{2}}=N C[Z(t)]$.

Since $Y^N$ has a multivariate normal distribution then every linear combination of its components is normally distributed. 
Therefore, given $B=[ b_{1} , b_{2}, \cdots , b_{|\Lambda|}]$ where $b_{1},b_{2},...,b_{|\Lambda|} \in \mathbb{Z}$, we can consider the random variable $BY^N(t)$, which defines a linear combination of the species at time $t$. For every $t$, $B{Y^{N}(t)}$ is a normal random variable, whose expected value and variance are
\begin{equation}
		E[B{Y^{N}(t)}]= BE[{Y^{N}}(t)]
		\label{eq:excom}
\end{equation}
\begin{equation}
		 C[B{Y^{N}(t)}]= BC[{Y^{N}}(t)]B^T
		\label{eq:varcom}
\end{equation}
For a specific time  $t_k$, it is possible to calculate the probability that $B{Y^{N}(t_k)} $ is within a set $I$ of closed, disjoint real intervals $[l_i,u_i]$, where $l_i,u_i \in \mathbb{R} \cup \{+\infty,-\infty\}$.
This probability $\Omega_{Y^N,B,I}(t_k)$ is given by  
\begin{equation}
		\Omega_{Y^N,B,I}(t_k) = \sum_{[l_i,u_i] \in I} \int\limits_{l_i}^{u_i} g(x|E[B{Y^{N}(t_k)} ],C[B{Y^{N}(t_k)}])dx
\label{eq:Combination}
\end{equation}
where $g(x|EV,\sigma^2)$ is the Gaussian distribution with expected value $EV$ and covariance $\sigma^2$. We recall that it is possible to find numerical solution of Eqn \eqref{eq:Combination} in constant time using the Z table \cite{patel1996handbook}.

\begin{example}
Consider the CRN of Example \ref{ex-1}, then we can obtain the probability that $\#\lambda_1-2\#\lambda_3$ is at least $10$ at time $20$ by defining $B'=[1,0,-2], \, I'=\{[10,+\infty]\}$ and calculating $\Omega_{Y^N,B',I'}(20)$. 
\end{example}
The following theorems are consequences of results in \cite{Wallace2012}, which can be generalized for reactions with a finite number of reagents and products. They show asymptotic pointwise convergence of expected value, variance and probability.

\begin{theorem}{}\label{prob-thm}
Let $C=(\Lambda,R)$ be a CRN. Suppose the solution of Eqn \eqref{eq:COV1} is bounded, then, approaching the thermodynamic limit, for any finite instant of time $t_i$
\begin{equation}
\lim_{N \to \infty} \| \Omega_{Y^N,B,I}(t_i)- \widetilde{\Omega}_{X^N,B,I}(t_i) \| = 0 , 
\end{equation}
where  $\widetilde{\Omega}_{X^N,B,I}(t_i)$ is the probability that $B(X^N)$ is within $I$ at time $t_i$.
\end{theorem}

\begin{theorem}{}\label{exp-var-thm}
Suppose the solution of Eqn \eqref{eq:COV1} is bounded, then, approaching the thermodynamic limit, for any finite instant of time $t_k$
\begin{equation}
\lim_{N \to \infty} \| C[BY^{N}(t_k)] - C[BX^{N}(t_k)] \| = 0
\end{equation}
\begin{equation}
\lim_{N \to \infty} \| E[BY^{N}(t_k)] -E[BX^{N}(t_k)] \| = 0.
\end{equation}
\addtocounter{theorem}{-2}
\end{theorem}

To solve the differential equations \eqref{eq:EV1} and \eqref{eq:COV1}, it is necessary to use a numerical method such as adaptive Runge-Kutta algorithm \cite{butcher1987numerical}. This yields the solution for a finite set of sampling times $\Sigma=[t_1,...,t_{|\Sigma|}] \in \mathbb{R}^{|\Sigma|}$, where $t_1\leq...\leq t_k  \leq ... \leq t_{|\Sigma|}$ and $|\Sigma|$ is the sample size. 
Assuming $Y^N$ is separable, that is, it is possible to completely define the behavior of $Y^N$ by only considering a countable number of points, we can calculate $\Omega_{Y^N,B,I}$ for any point in $\Sigma$ and if points are dense enough then this set exhaustively describes the probability that $BX^N$ is within $I$ over time.
This restriction is not a limitation since for any stochastic process there exists a separable modification of it \cite{ito2006essentials}.

\section{Stochastic Evolution Logic (SEL)}\label{logic-sec}
Let $C = (\Lambda,R)$ be a CRN with initial state $x_0$, in a system of size $N$. We now define the logic SEL (Stochastic Evolution Logic) which enables evaluation of the probability, variance and expectation of linear combinations of populations of the species of $C$.

The syntax of SEL is given by 

\[
\eta := P_{\sim p}[ B,I]_{[t_1,t_2]}  \quad |\quad Q_{\sim v}[B]_{[t_1,t_2]} \quad |\quad \eta_1 \wedge \eta_2 \quad  | \quad \eta_1 \vee \eta_2 \quad
\]

\noindent
where $Q=\{supV, infV,supE,infE\}$, $\sim=\{<,>\}$, $p \in [0,1]$, $v \in \mathbb{R}$, $B \in \mathbb{Z}^{|\Lambda|}$, $I=\{[l_i,u_i] \, | \, l_i,u_i \in \mathbb{R} \cup [+\infty,-\infty] \, \wedge [l_i,u_i]\cap [l_j,u_i]=\emptyset, \, i \neq j\}$ and $[t_1,t_2]$ is a closed interval, with the constraint that $t_1\leq t_2$ and $t_1,t_2 \in \mathbb{R}$. If $t_1=t_2$ the interval reduces to a singleton.

Formulae $\eta$ describe global properties of the stochastic evolution of the system. 
$(B,I)$ specifies a linear combination of the species of $C$ and a set of intervals, where $B \in \mathbb{Z}^{|\Lambda|}$ is the vector defining the linear combination and $I$ represents a set of disjoint closed real intervals. $P_{\sim p}[ B,I]_{[t_1,t_2]}$ is the probabilistic operator, which specifies the probability that the linear combination defined by $B$ falls within the range $I$ over the time interval $[t_1,t_2]$. $supE,infE,infV,supV$ respectively yield the supremum and infimum of expected value and variance of the random variables associated to $B$ within the specified time interval.

\begin{example}
Consider the CRN of Example \ref{ex-1}. Checking if the variance of $\#\lambda_1$ remains smaller than $K_1$ within $[t_j,t_k]$ can be expressed as  $supV_{<K_1}[[1,0,0]]_{[t_j,t_k]}$.  Another example is checking if, in the same interval, $(\#\lambda_1-\#\lambda_2)$ is at least $K_2$ or within $[K_3,K_4]$, with $K_3<K_4<K_2$, with probability greater than $0.95$: \allowbreak{$P_{> 0.95}[[1,-1,0], ([K_3,K_4],[K_2,\infty])]_{[t_j,t_k]}$}. Equivalently, instead of writing $B$, we write directly the linear combination it defines. For example, in the latter case we have $P_{> 0.95}[(\#\lambda_1-\#\lambda_2), ([K_3,K_4],[K_2,\infty])]_{[t_j,t_k]}$.
\end{example}

\subsubsection{Semantics}
Given a CRN $C=(\Lambda,R)$ with initial configuration $x_0$ in a system of fixed volumetric factor $N$, its stochastic behaviour is described by the CTMC $X^N$ of Definition \ref{ctmc-defn}.
We define a path of CTMC $X^N$ as a sequence $\omega=x_0 t_1 x_1 t_1 x_2...$ where $x_i$ is a state and $t_i \in \mathbb{R}_{>0} $ is the time spent in the state $x_i$. A path is finite if there is a state $x_k$ that is absorbing. $\omega \otimes t$ is the state of the path at time $t$. $Path(X^N,x_0)$ is the set of all (finite and infinite) paths of the CTMC starting in $x_0$. We work with the standard probability measure $Prob$ over paths $Path(X^N,x_0)$ defined using cylinder sets \cite{Kwiatkowska2007}.

We first define when a path $\omega$ satisfies $ (B,I)$ at time $t$ 
\[
\omega,t \models  (B,I) \quad    \leftrightarrow  \quad   \exists [l_i,u_i] \in I \,.\,   l_i \leq B(\omega \otimes t) \leq u_i .
\]
\noindent
Note that $B(\omega \otimes t)$ is well defined because $\omega \otimes t \in \mathbb{N}^{|\Lambda|}$.
We now define $Pr^{X^N}_{B,I} (t)=Prob\{\omega \in Path(X^N,x_0)\,|\, \omega ,t \models  (B,I)\}$, then if the time interval is a singleton the satisfaction relation for the probabilistic operator is
\[
{X^N}, x_0 \models P_{\sim p}[ B,I]_{[t_1,t_1]}   \quad \leftrightarrow \quad  Pr^{X^N}_{B,I} (t_1) \sim p
\]
Instead, for $t_1 < t_2$ we have
\[
{X^N}, x_0 \models P_{\sim p}[ B,I]_{[t_1,t_2]}   \quad \leftrightarrow \quad \frac{1}{t_2-t_1}  \int_{t_1}^{t_2} Pr^{X^N}_{B,I} (t)\,\mathrm{d}t \sim p
\]
 $Pr^{X^N}_{B,I} (t)$ is the probability of the set of paths of $X^N$ such that the linear combination of the species defined by $B$ falls within $I$. It is well defined since we have previously defined the probability measure $Prob$ on $Path(X^N, x_0)$.  To define the satisfaction relation of the probabilistic operator we simply take the average value of $Pr^{X^N}_{B,I} (t)$ during the interval $[t_1,t_2]$. For the remaining operators the satisfaction relation is defined as
\[
{X^N},x_0 \models supV_{\sim v}[B]_{[t_1,t_2]}  \quad    \leftrightarrow \quad   sup(C[B(X^N)],[t_1,t_2]) \sim v
\]
\[
{X^N},x_0 \models infV_{\sim v}[B]_{[t_1,t_2]}  \quad    \leftrightarrow \quad     inf(C[B(X^N)],[t_1,t_2]) \sim v
\]
\[
{X^N},x_0 \models supE_{\sim v}[B]_{[t_1,t_2]}  \quad    \leftrightarrow \quad     sup(E[B(X^N)],[t_1,t_2]) \sim v
\]
\[
{X^N},x_0 \models infE_{\sim v}[B]_{[t_1,t_2]}  \quad    \leftrightarrow \quad     inf(E[B(X^N)],[t_1,t_2]) \sim v
\]
\[
{X^N}, x_0 \models \eta_1 \wedge \eta_2 \quad    \leftrightarrow  \quad     {X^N}, x_0 \models \eta_1 \wedge {X^N}, x_0 \models \eta_2 
\]
\[
{X^N}, x_0 \models \eta_1 \vee \eta_2  \quad    \leftrightarrow \quad     {X^N}, x_0 \models \eta_1 \vee {X^N}, x_0 \models \eta_2 
\]
 $inf(\cdot,[t_1,t_2])$ and $sup(\cdot,[t_1,t_2])$ respectively denote the infimum and supremum within $[t_1,t_2]$. 

\subsection{LNA-based Approximate Model Checking for CRNs}
Stochastic model checking of CRNs is usually achieved by transient analysis of the CTMC $X^N$ \cite{Kwiatkowska2007}, which involves solving the CME and thus suffers from the state-space explosion problem. 
We propose an approximate model checking algorithm based on LNA.
The inputs are a SEL formula $\eta$, the stochastic process $X^N$ induced by the CRN and initial state $x_0$. 
The output is $true$ in case the formula is verified, and otherwise $false$. 


The algorithm proceeds by induction on the structure of formula $\eta$, successively computing whether each subformula is satisfied or not.
%
%
We assume that  Eqn \eqref{eq:EV1} and \eqref{eq:COV1} are solved numerically where $\Sigma$ is the finite set of sample points on which their solution is defined and that $t_0$, initial time, and $t_{max}$, final time, are always sampling points.


\subsubsection{Probabilistic operator.}
To evaluate $P_{\sim p}[ (B,I)]_{[t_1,t_2]}$ we construct the function $Prob_{ (B,I)}(t) = \Omega_{Y^N,B,I}(t_i)$ for $ t \in [t_i,t_{i+1}), t_i,t_{i+1} \in \Sigma $ (alternatively, can be constructed as the interpolation of the values of $\Omega_{Y^N,B,I}$ over $\Sigma$ points). 

\begin{lemma}
$Prob_{ (B,I)}$ is integrable on $\mathbb{R}_{\geq 0}$.
\addtocounter{lemma}{-1}
\end{lemma}
Theorem \ref{prob-thm} guarantees the pointwise correctness of $Prob_{ (B,I)}$ and its integrability allows us to compute the following approximation, then compare to threshold $p$ to decide the truth value.
If $t_2 \neq t_1$ then $\frac{1}{t_2-t_1} \int_{t_1}^{t_2} Pr^{X^N}_{B,I} (t)\,\mathrm{d}t  \approx \frac{1}{t_2-t_1}\int_{t_1}^{t_2}Prob_{B,I}(t)dt$ else if $t_1=t_2$  then  $Pr^{X^N}_{B,I} (t_1)  \approx Prob_{B,I}(t_1)$.

\subsubsection{Expectation and variance operators}
  
To evaluate $ sup(C[B(X^N)],[t_1,t_2]) $, $inf(C[B(X^N)],[t_1,t_2])$, $ sup(E[B(X^N)],[t_1,t_2]) $ and $inf(E[B(X^N)],[t_1,t_2])$  we use the LNA, namely, compute the expected value and variance of Eqn \eqref{eq:varcom} and \eqref{eq:excom}. Theorem \ref{exp-var-thm} guarantees the quality of the approximation.
We can now compute the following approximations, then compare to the threshold $v$:
\[
sup(C[B(X^N)],[t_1,t_2])\approx max  \{ C[BY^{N}(t_k)] \, | \, (t_k \in \Sigma \wedge t_1\leq t_k \leq t_2) \vee (t_k \in L_{[t_1,t_2]} )\} 
\]
\[
inf(C[B(X^N)],[t_1,t_2])\approx min \{ C[BY^{N}(t_k)] \, | \, (t_k \in \Sigma \wedge t_1\leq t_k \leq t_2) \vee (t_k \in L_{[t_1,t_2]} )\}
\]
and similarly for the expected value. $L_{[t_1,t_2]}=\{t_i | t_i \in \Sigma  \,  \, \wedge \,\, \nexists t_j \in \Sigma$ $such$ $that$ $|t_1-t_j|<|t_1-{t}_{i}|  \}$ ensures that for any time interval there is at least one sampling point, even if the interval is a singleton.  
Note that, for each sub-formula, the algorithm involves the calculation of some quantity, so one can define a quantitative semantics for SEL as in \cite{donze2010robust}.


LNA-based model checking can also be used for systems far from the thermodynamic limit, at a cost of some loss of precision. 
LNA assumes continuous state space, and it is not possible to justify this assumption for very small populations.
However, if the distributions of interest are not multi-modal and the noise term is finite and approximated by a Gaussian distribution, then LNA gives very good approximation even for quite small systems. It is clear that model checking accuracy increases as $N$ grows. 
We emphasise that the model checking algorithm we have presented is also able to handle CRNs whose stochastic semantics is an infinite CTMC, which occur frequently in biological models.


\subsubsection{Complexity of LNA-based approximate model checking}
The time complexity for model checking formula $\eta$ against a CRN $C=(\Lambda,R)$ is linear in $|\eta|$. 
In the worst case, analysis of a single operator requires the solution of $O(|\Lambda|^2)$ polynomial differential equations for a bounded time. However, an efficient implementation can solve the $O(|\Lambda|^2)$ ODEs only once for the interval $[0,t_{max}]$, and then reuse this result for every operator, where $t_{max}$ is the greatest (finite) time of interest. 
Note that ODEs are solved in terms of concentrations (a value between $0$ and $1$ by convention), ensuring independence of the number of molecules of each species, although stiffness can slow down the solution of the LNA. 

\section{Experimental Results}\label{results-sec}

We implemented the methods in a framework based on Matlab and Java. The experiments were run on an Intel Dual Core $i7$ machine with $8$ GB of RAM. To solve the differential equations, we use $Matlab$ $ode45$, a variable step Runge-Kutta algorithm.
We employ LNA-based model checking for the analysis of four biological reaction networks: a Phosphorelay Network \cite{Csikasz-Nagy2011}, a Gene Expression Model \cite{thattai2001,Henzinger2011}, the FGF pathway \cite{Heath2006} and the GW network \cite{cardelli2014morphisms}. For every network, the CRN and parameters have been taken from the referenced papers. 
We coded the same CRNs in PRISM in order to compare accuracy and time of execution with standard uniformisation of the CME \cite{Kwiatkowska2007} and statistical model checking (SMC) techniques (confidence interval method) as implemented in PRISM. 
For the FGF and GW case studies, we cannot use global analysis nor SMC, because the state space is too large for direct analysis, and SMC requires many time-consuming simulations to obtain good accuracy. 

\textbf{Phosphorelay Network.}
We consider a three-layer phosphorelay network whose structure is derived from \cite{Csikasz-Nagy2011}. Each layer $(L1,L2,L3)$ can be found in phosphorylate form $(L1p,L2p,L3p)$. We consider the initial condition $ \#L1p=\#L2p=\#L3p=0, \, \#L1=\#L2p=\#L3p=Init$,  where $Init \in \mathbb{N}$. Then we analyse the ligand $B$, whose initial condition is $\#B=3*Init$.  We are interested in checking the following SEL property:
\[
P_{>0.7}[(\#L1p-\#L3p),[0,+\infty]]_{[0,100]}\wedge P_{>0.98}[(\#L3p-\#L1p),[0,+\infty]]_{[300,600]}
\]
which is verified if, in the first interval, the probability that $\#L1p$ is greater than $\#L3p$ is $>0.7$ and if, between $300$ and $600$, with probability $> 0.98$, $\#L3p$ is greater than $\#L1p$.
We evaluate this formula in three different initial conditions, firstly $Init=32$ and $N=5000$, then $Init=64$ and $N=10000$, and finally $Init=100$ and $N=15625$, so the same concentration but different numbers of molecules.
In all cases, the LNA-based model checking evaluates the formula as true. To understand the quality of the approximation, we check the following quantitative formula $P_{=?}[(\#L3p-\#L1p,[0,+\infty])]_{[T,T]}$ for $ T \in [0,600]$ (in our implementation  $=?$ gives the quantitity calculated by model checking the operator). We compare the results with the evaluation of the corresponding CSL formula using standard uniformisation ($Unif$) with error $10^{-7}$ \cite{Kwiatkowska2007}. The following table shows the results. $MaxErr$ is the maximum error computed by LNA-based approach compared to standard uniformisation and $AvgErr$ is the average error; $Time(\cdot)$ stands for execution time.

\begin{center}
	\begin{tabular}{|l|l|l|l|l|}
	\hline
	Init                 &  \,		Time (LNA)	 \,	 & \,  	Time (Unif)	\, 	  & \,  MaxErr \, 	  & \,  AvgErr    \\ 
	\hline
	20        & 		  0.22 sec           	 & 	 	2 min   &  0.0675   &  0.0519 \\ 
	32   &  	  0.23   sec      	 & 	 	5 min     	&  0.059   &  0.02        \\ 
	64   & 		  0.26   sec      	 & 	 	$>$ 2 hr     &  0.0448  &  0.0027  \\ 
	100   & 		  0.3   sec      	 & 	 	$>$ 2 hr     &  0.03  &  0.0011  \\ 
	\hline
	\end{tabular}
\end{center}
Note that as $Init$ increases the error of our method decreases, while the execution time is practically independent of the molecular count.
LNA-based algorithms are faster in all cases. Thus our approach can be used even for quite small population systems, giving a fast approximate stochastic characterization. 

\textbf{Gene Expression.}
We consider a simple CRN that models the transcription of a gene into an mRNA molecule, and the translation of the latter into a protein.
The CRN, rates and initial conditions are the same as in \cite{Henzinger2011}. The stochastic semantics of the reaction network is an infinite CTMC, and we use this model to show that our method can handle infinite state-space processes.
We consider the quantitative property $supE_{=?}[\#mRNA]_{[T,T]}$, which gives the number of molecules of $mRNA$ in the system at time $T$.
We compare our method with SMC estimation of the same property by using $50000$ simulations, for $T=\{ 300,600,900,1200\}$, and in the following tables we compare the results in terms of execution time ($Time (\cdot)$) and expected value of $\#mRNA$ estimated ($ExpVal (\cdot)$). 
LNA-based model checking is several orders of magnitude faster without loss of accuracy.
\begin{center}
	\begin{tabular}{|l|l|l|l|l|}
	\hline
	T                &  \,		Time (LNA)	 \,	 & \,  	Time (Simul)	\, 	  & \,  ExpVal (LNA) \, 	  & \,  ExpVal (Simul)    \\ 
	\hline
	300        & 		  0.52 sec           	 & 	 	75 sec   &  100.17   &  100.14 $\pm$ 0.1 \\ 
	600   &  	  0.54   sec      	 & 	 	198 sec     	&  142.15   &  142.11 $\pm$ 0.1       \\ 
	900   & 		  0.54   sec      	 & 	 	337 sec     &  159.73  &  159.74 $\pm$ 0.1  \\ 
	1200   & 		  0.56   sec      	 & 	 	483 sec     &  167.1  &  167.1 $\pm$ 0.1  \\ 
	\hline
	\end{tabular}
\end{center}

\textbf{FGF.}
We consider the model of Fibroblast Growth Factor (FGF) signalling pathway developed in \cite{Heath2006} composed of more than $50$ reactions and species.
We consider the system with initially 105 molecules for species with non-zero initial concentration.
Analysis of the model reveals that the phosphorylated form of $FRS2$ can bind the protein $Src$, and then this new complex, $Src$:$FRS2$, can relocate out. We want to check if the expected value of  $\#Src$:$FRS2$ during the first $3000$ seconds reaches a maximum value greater than $40$. We do that by checking the property $supE_{>40}[\#Src$:$FRS2]_{[0,3000]}$. The formula evaluates to true, and in Figure \ref{fig:sto} we analyze the expected value and standard deviation of $ \#Src$:$FRS2$. We obtain these values directly from the logic considering the quantitative interpretation of $supE_{=?}[\#Src$:$FRS2]_{[T,T]}$ and $supV_{=?}[\#Src$:$FRS2]_{[T,T]}$ for $T\in [0,3000]$. It is possible to see that, after an initial peak, relocation causes exponential decay.  

\begin{wrapfigure}{r}{0.54\textwidth}
  \vspace{-4em}
  \begin{center}
\includegraphics[scale=0.10]{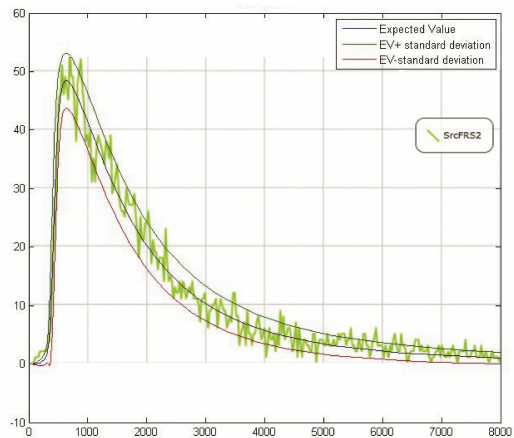}
  \end{center}
    \vspace{0em}
  \caption{Expected number and standard deviation of species of $\#Src$:$FRS2$ in the FGF pathway during the first $8000$ seconds estimated by our method is compared with a stochastic simulation of the same species.}
    \label{fig:sto}
   \vspace{-1.5em}
\end{wrapfigure}

In the same figure we show a single stochastic simulation of the system for the same initial conditions, confirming our evaluation. Moreover, the approximation can be justified theoretically. 
$\#Src$:$FRS2$ converges to zero necessarily and this demonstrates the unimodality of the distribution of the species; we note that the variance is finite, so 
Eqn \eqref{eq:hypothsis} holds.

\textbf{DNA strand displacement of GW network}
GW is a network related to the G2-M cell cycle switch \cite{novak1993numerical}. Under particular initial conditions, it has been shown that GW can emulate the Approximate Majority algorithm \cite{cardelli2014morphisms}. Here, we consider the two-domain DNA strand-displacement implementation of GW \cite{Cardelli2010}. The corresponding CRN is composed of $340$ species and $240$ reactions. For our analysis the species of interest are $R$ and $P$, whose initial conditions are $\#R=90$ and $\#P=10$; initial conditions of other species are taken from the referenced papers. We check the property $P_{>0.9}[\#R-\#P,[50,+\infty]]_{[6000,35000]}$ for a system of size $N=45000$, which is verified as true in $28$ minutes.
 
\section{Concluding Remarks}\label{concl-sec}
We presented a novel probabilistic logic for analysing stochastic behaviour of CRNs
and proposed an approximate model checking algorithm based on the LNA of the CME. 
We have demonstrated on four non-trivial examples that
LNA-based model checking enables analysis of CRNs with hundreds of species, and even infinite CTMCs, at a cost of some loss of accuracy. 
It would be interesting to find bounds on the approximation error when the system is far from the thermodynamic limit. However, the error is not only dependent on the value of $N$, but also on the structure of the CRN, the rates, and the property.
As future work, we plan to improve the accuracy of the method near critical points similarly to the approach of \cite{Elf2003}, and to extend the logic with more expressive temporal operators in the style of CSL \cite{aziz2000model}. We also intend to release a software tool based on LBS \cite{Pedersen2010}. 

\section*{Acknowledgements}
The authors would like to thank Luca Bortolussi for helpful discussions. 

\bibliographystyle{abbrv}
{\scriptsize

\bibliography{LanguagesForBiology.bib}
}

\end{document}